\documentclass[conference]{IEEEtran}
\usepackage[utf8]{inputenc}
\usepackage[T1]{fontenc}
\usepackage{amsmath,amsthm,amssymb,amsfonts}
\usepackage{algorithmic}
\usepackage{graphicx}
\usepackage{hyperref}
\usepackage{xcolor}
\usepackage{caption}
\usepackage{subcaption}
\usepackage{cite}
\usepackage{booktabs}
\usepackage{multirow}
\usepackage{array}
\usepackage{float}
\usepackage{stfloats}
\usepackage{braket}

\newtheorem{lemma}{Lemma}
\newtheorem{theorem}{Theorem}

\newtheorem{corollary}{Corollary}

\title{\LARGE \bf Quantum Surrogate-Driven Image Classifier: A Gradient-Free Approach to Avoid Barren Plateaus}

\begin{document}

\author{Yichen Xie \\ La Salle College, HKSAR \\ s21325@lsc.hk}

\maketitle

\begin{abstract}
Training deep quantum neural networks (QNNs) for image classification is notoriously difficult due to vanishing gradients (barren plateaus) and limited nonlinearity in purely unitary circuits. We propose a novel gradient-free surrogate-driven framework combined with mid-circuit measurement and reset of ancillary qubits to induce effective nonunitarity. Our approach uses a classical neural surrogate to predict measurement outcomes from circuit parameters to avoid direct gradients. Theoretical results prove that bypassing quantum gradients mitigates plateau issues. Experiments on MNIST, CIFAR-10, and CIFAR-100 with 15-qubit, 6-layer circuits using four resettable ancillas demonstrate superior accuracy compared to direct-gradient QNNs and classical baselines. Our method also serves as a potential for a generalized training framework applicable to various QNN architectures beyond image classification.
\end{abstract}

\section{Introduction}

Parametric quantum circuits (PQCs) leverage the exponential dimensionality of an \(n\)-qubit Hilbert space, offering powerful potential for tasks such as image classification. However, their practical training faces critical challenges: the barren plateau phenomenon, causing gradients to vanish exponentially with increasing circuit depth or qubit count, severely impeding gradient-based optimization \cite{mcclean2018barren,Cerezo2021}; and the inherent linearity of unitary operations within PQCs, restricting their expressibility compared to classical deep networks that utilize nonlinear activations \cite{Mitarai2018,Preskill2018}.

To address these challenges, we propose a novel surrogate-based optimization approach. Instead of directly computing quantum gradients, we train a classical neural network surrogate model to learn the mapping from circuit parameters to measurement outcomes, thus providing efficient surrogate gradients without quantum differentiation \cite{Benedetti2019,Schuld2015}. Furthermore, we introduce mid-circuit measurement and reset operations on designated ancilla qubits after each circuit layer, creating branching effects analogous to nonlinear activations in classical neural networks, and facilitating resource-efficient qubit reuse \cite{Chertkov2022,IBM2021}.

Additionally, recognizing that naive amplitude encoding is inefficient for large images, we first compress input images into a reduced feature space (e.g., 256 for MNIST, 512 for CIFAR), then use a parameter-projection network to map these compressed features into an extensive set of circuit parameters \cite{Havlicek2019,Lloyd2020,Schuld2021}. Our integrated approach, combining surrogate-driven optimization, nonlinear-like mid-circuit operations, and efficient data encoding, demonstrates superior performance compared to direct-gradient quantum neural networks and classical baselines, as validated by extensive simulations on benchmark datasets including MNIST, CIFAR-10, and CIFAR-100 \cite{Schuld2015,Preskill2018,Schuld2021}.

\section{Architecture and Theoretical Framework}

\subsection{Overview}

\begin{figure}[h] 
\centering
\includegraphics[width=\linewidth]{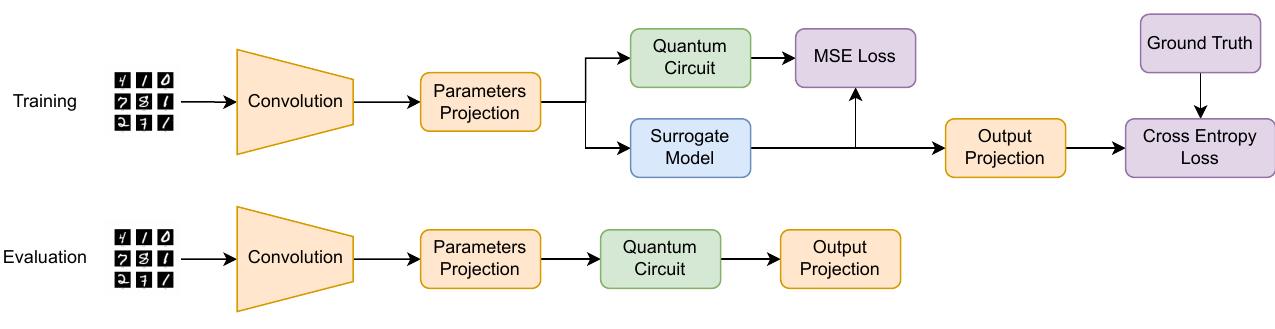}
\caption{Framework of the architecture during training and evaluation.}
\label{fig:surrogate}
\end{figure}

We now describe the overall architecture. The quantum circuit has $n=15$ qubits, including $n_a=4$ ancillas at wires 3,6,9,12, across $L=6$ layers. Each layer applies parameterized unitaries, measures ancillas, resets them to $|0\rangle$ if measured as $|1\rangle$, and proceeds to the next layer. A simplified 6-qubit, 3-layer version is shown in Figure~\ref{fig:surrogate_circuit}.

\begin{figure}[h]
\centering
\includegraphics[width=0.75\linewidth]{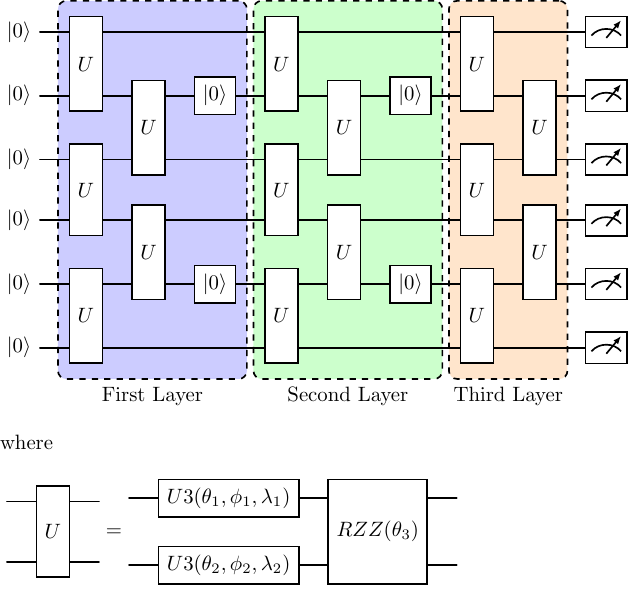}
\caption{Simplified quantum circuit with 6 qubits, 3 layers, and ancillas at wires 1,5. Each U module has 7 parameters.}
\label{fig:surrogate_circuit}
\end{figure}

For MNIST ($28\times28=784$ pixels), images are compressed to 256 features; for CIFAR-10/100 ($32\times32\times 3=3072$), they are compressed to 512 features, using a simple VGG-style CNN $C_w(x)$ consisting of four blocks with Batch Normalization, ReLU activations \cite{householder1941, nair2010rectified}, and max pooling. The compressed vector $\mathbf{z}\in\mathbb{R}^{256}$ or $\mathbb{R}^{512}$ is then fed into a two-layered MLP $\Gamma_w(\mathbf{z})\in\mathbb{R}^p$, generating $p$ trainable angles $\boldsymbol{\theta}$ (with $p=588$ for 6 layers and 15 qubits). By measuring ancillas after each layer and resetting them, the main qubits undergo an effective nonlinear transformation (Theorem~\ref{thm:nonlinearity}).

The circuit output (expectation value with Pauli-Z observable) is passed into a projection layer $O_w(x)\in[0,1]$ for classification. We compute the cross-entropy (CE) loss:
\begin{equation}
\mathcal{L}_1(x,y)=-\sum_{j=1}^{C}y_j\log\bigl(O_w(S_w(\Gamma_w(C_w(x_j))))\bigr)
\end{equation}
where $S_w$ is the surrogate model, $x$ is the input, $y_j$ is the one-hot label, and $C$ is the number of classes. Gradients $\nabla O_w$ are obtained by backpropagating $\nabla\mathcal{L}_1$.

To avoid gradients through the quantum circuit during training, we sample parameter points around $\boldsymbol{\theta}_t$, run the circuit, and fit a surrogate neural network $S_w(\boldsymbol{\theta})\approx\mathbf{m}(\boldsymbol{\theta})$ using mean-squared error (MSE):
\begin{equation}
\mathcal{L}_2(\boldsymbol{\theta}) 
= 
\frac{1}{2p}\sum^{p}_{i=1}\bigl(S_w(\boldsymbol{\theta_i})-\mathbf{m}(\boldsymbol{\theta_i})\bigr)^2.
\end{equation}
This surrogate provides differentiable approximations of quantum outcomes, allowing efficient classical updates to CNN and MLP parameters.

Traditional QNNs apply purely unitary layers, measuring only at the end, resulting in linear transformations and susceptibility to barren plateaus \cite{mcclean2018barren}. Our method introduces mid-layer ancilla measurements, inducing nonunitary transformations, and circumvents barren plateaus by using gradient-free surrogate modeling.

\subsection{Nonunitary Gates by Measuring and Resetting Ancillas}

We first prove that measuring and resetting a subset of qubits can implement a nontrivial class of nonlinear maps on the remaining subsystem without collapsing the entire system to a single state.

\begin{lemma}[Non‑unitary but Non‑collapsing Layer]
\label{lemma:nonunitary}
Let $n=n_m+n_a$ qubits be divided into $n_m$ \emph{main} qubits and
$n_a$ \emph{ancillas} ($n_a<n$).
Start from an $n$–qubit density operator $\rho$.
Apply a global unitary $U(\boldsymbol\theta)$, then measure every ancilla
in the $\{\lvert0\rangle,\lvert1\rangle\}$ basis.
For each outcome $x\in\{0,1\}^{n_a}$ \emph{reset} the ancillas to
$\lvert0^{\otimes n_a}\rangle$.
Denote the overall completely‑positive trace‑preserving (CPTP) map by
$\mathcal M$.
Tracing out the ancillas yields the channel
\begin{equation}
   \operatorname{Tr}_{\!\mathrm{anc}}\!\bigl[\mathcal M(\rho)\bigr]
   \;=\;
   \sum_{x\in\{0,1\}^{n_a}}\! K_x\,\rho\,K_x^{\dagger},
\end{equation}
where each Kraus operator on the main subsystem is
\begin{equation}\label{eq:kraus-main}
   K_x
   \;=\;
   {}_{\mathrm{anc}}\!\!\bra{0^{\otimes n_a}}
          \bigl(\,\lvert0^{\otimes n_a}\rangle\!\langle x|\;\otimes I_{n_m}\bigr)
          U
   \;=\;
   {}_{\mathrm{anc}}\!\!\bra{x}\,U.
\end{equation}
Because the $K_x$ are generically non‑unitary and (for entangling $U$)
mutually distinct, $\mathcal M$ performs a nonlinear, non‑collapsing
transformation on the main qubits that cannot be written as a single
unitary acting only on those qubits.
\end{lemma}

\begin{proof}
Let
$
   M_x
   =\bigl(\lvert x\rangle\!\langle x|\bigr)_{\!\mathrm{anc}}\!
    \otimes I_{n_m}
$
project onto the ancilla outcome $x$.
Define the \emph{reset operator}
\begin{equation}
   R_x
   \;=\;
   \bigl(\,\lvert0^{\otimes n_a}\rangle\!\langle x|\bigr)_{\!\mathrm{anc}}
   \otimes I_{n_m},
   \qquad
   x\in\{0,1\}^{n_a}.
\end{equation}
For each branch $x$ we evolve as
$
   \rho \;\mapsto\; R_x\,M_x\,U\,\rho\,U^{\dagger}M_x\,R_x^{\dagger}.
$
Summing over outcomes gives the global CPTP map
\begin{equation}
   \mathcal M(\rho)
   \;=\;
   \sum_{x} R_x M_x U\,\rho\,U^{\dagger}M_x R_x^{\dagger}.
\end{equation}

Tracing out the ancillas and inserting the resolutions of the identity,
\begin{equation}
\begin{aligned}
   &\operatorname{Tr}_{\!\mathrm{anc}}\!\bigl[\mathcal M(\rho)\bigr] \\
   &\,=\;
     \sum_{x}
        \operatorname{Tr}_{\!\mathrm{anc}}
        \Bigl[
          \bigl(R_xM_xU\bigr)\rho
          \bigl(U^{\dagger}M_xR_x^{\dagger}\bigr)
        \Bigr]                                       \\[2pt]
   &\,=\;
     \sum_{x}
       \bigl({}_{\mathrm{anc}}\!\bra{0^{\otimes n_a}}
             R_xM_xU\bigr)\,
       \rho\,
       \bigl(U^{\dagger}M_xR_x^{\dagger}
             \ket{0^{\otimes n_a}}_{\mathrm{anc}}\bigr) \\[2pt]
   &\,=\;
     \sum_{x} K_x\,\rho\,K_x^{\dagger},
\end{aligned}
\end{equation}
with $K_x$ defined in~\eqref{eq:kraus-main}.  The second equality
uses the fact that
$
   (\lvert y\rangle\!\langle y|)_{\!\mathrm{anc}}
$
appearing in the partial trace selects $y=0^{\otimes n_a}$
because $R_x$ resets ancillas to that state.

Unitarity of $U$ and orthogonality of the projectors imply
$\sum_x K_x^{\dagger}K_x=I_{n_m}$, so the set
$\{K_x\}_{x\in\{0,1\}^{n_a}}$
constitutes a valid Kraus representation of a CPTP map on the main
qubits.

If $U$ entangles ancillas with main qubits, the operators $K_x$ differ
for different $x$ and are not proportional to unitaries on the main
subsystem.  Multiple Kraus terms therefore survive in the sum,
demonstrating that $\mathcal M$ is generally a non‑unitary,
non‑projective transformation that nonetheless leaves the main qubits
in a mixed state retaining information from all measurement
branches.
\end{proof}

\begin{corollary}[Ancilla reuse after measure--reset]
\label{cor:ancilla-reuse}
Using notation from Lemma~\ref{lemma:nonunitary}, the full post-measurement-reset state factorizes as
\begin{equation}
\mathcal M(\rho) = |0^{\otimes n_a}\rangle\!\langle 0^{\otimes n_a}|_{\mathrm{anc}} \otimes \sum_{x\in\{0,1\}^{n_a}} K_x\rho K_x^{\dagger}.
\end{equation}
Thus, the ancillas are deterministically reset to $|0^{\otimes n_a}\rangle$, remain unentangled from the main register, and can be safely reused.
\end{corollary}

\begin{proof}
Start from $\mathcal M(\rho)=\sum_{x}(R_xM_xU)\rho(U^{\dagger}M_xR_x^{\dagger})$, substitute $R_xM_x=(|0^{\otimes n_a}\rangle\!\langle x|)\otimes I_{n_m}$, and insert identity resolution $I_{\mathrm{anc}}=\sum_{y}|y\rangle\!\langle y|$. Only $y=0^{\otimes n_a}$ remains, yielding the factorization with $K_x={}_{\mathrm{anc}}\langle x|U$.
\end{proof}

\begin{theorem}[Nonlinearity and Preserved Expressibility]
\label{thm:nonlinearity}
The measure-and-reset operation from Lemma~\ref{lemma:nonunitary} is inserted after each layer in a deep circuit of total depth $L$. Let $\Phi(\boldsymbol{\theta})$ be the resulting overall map from initial $n_m$-qubit states to final $n_m$-qubit states. Then for sufficiently large $n_a\ge 1$ and entangling unitaries $U_\ell(\boldsymbol{\theta}_\ell)$, the map $\Phi$ can realize a broad class of nonlinear transformations on the main qubits, while avoiding deterministic collapse onto a single pure state. Formally, there exist Kraus representations:
\begin{equation}
\mathrm{Tr}_{\mathrm{anc}} \Bigl[\,\prod_{\ell=1}^{L}\Bigl(R \circ \mathcal{M}_\ell\circ U_\ell(\boldsymbol{\theta}_\ell)\Bigr)\Bigr]
=
\sum_{k} V_k\,(\cdot)\,V_k^\dagger
\end{equation}
with $V_k$ nontrivial, and no single Kraus operator dominates all inputs. Hence, this multi-layer measure-and-reset design acts as an effectively nonlinear feedforward QNN, which can approximate a wide range of channel mappings without collapsing distinct inputs into a single final state.
\end{theorem}

\begin{proof} We recursively apply Lemma~\ref{lemma:nonunitary}. Each layer $\ell$ has a global unitary $U_\ell$ on $n$ qubits, followed by measuring ancillas plus reset. The partial trace over ancillas after $L$ layers yields $\sum_k V_k (\cdot) V_k^\dagger$ with $V_k$ formed by branching sequences of $M_x$ projectors. If the unitaries $U_\ell(\boldsymbol{\theta}_\ell)$ sufficiently entangle ancillas with main qubits, measurement outcomes condition distinct Kraus branches. Because resets restore the ancillas to $|0\cdots0\rangle$, each layer reintroduces nonlinearity. Non-collapse follows from the fact that no single outcome branch has unit probability for all states (unless gates are trivial). Thus multiple Kraus operators remain active, preserving variation among inputs. \end{proof}

\subsection{Surrogate Model for Gradient-Free Parameter Updates}

Next, we outline a gradient-free procedure to train these multi-layer measure-and-reset circuits. Denote the circuit parameters by $\boldsymbol{\theta}$, and let the final measured probabilities on $n_m$ qubits for a classical label $y$ be $p(y|\boldsymbol{\theta};x)$, where $x$ is the input data. In classification, we aim to minimize a cost function:
\begin{equation}
\begin{aligned}
\mathcal{L}(\boldsymbol{\theta})
=
\frac{1}{N}\sum_{(x_i,y_i)}\ell\bigl(p(\cdot|\boldsymbol{\theta};x_i),y_i\bigr),
\end{aligned}
\end{equation}
which is a cross-entropy in our case. A typical gradient-based approach would compute $\partial \mathcal{L}/\partial \theta_j$ via parameter-shift rules, but for large $n$ or deep $L$, these gradients can vanish due to barren plateaus\cite{mcclean2018barren}.

Instead, we train a classical neural surrogate $S(\boldsymbol{\theta})$ to predict the measurement outcomes of the quantum circuit (not the final loss). Specifically, let the circuit output a vector $\mathbf{m}(\boldsymbol{\theta},x)\in\mathbb{R}^d$ of measurement statistics (e.g.\ log-odds for each class), and let $S$ be a neural network that takes $(\boldsymbol{\theta}, x)$ as input and estimates $\mathbf{m}(\boldsymbol{\theta}, x)$. We train $S$ by sampling $\boldsymbol{\theta}$ in a neighborhood of the current parameter and measuring the quantum device on data $x$ to get true outcomes $\mathbf{m}$. Then $S$ is optimized (in a purely classical sense) to fit these measured samples:
\begin{equation}
\min_{w} \;\sum_{k} \bigl\|\mathbf{m}(\boldsymbol{\theta}^{(k)},x^{(k)}) - S_w(\boldsymbol{\theta}^{(k)},x^{(k)})\bigr\|^2.
\end{equation}
Once $S$ accurately approximates $\mathbf{m}$ near the current $\boldsymbol{\theta}$, we compute the classification loss purely classically:
\begin{equation}
\begin{aligned}
\widetilde{\mathcal{L}}(\boldsymbol{\theta})
&=
\frac{1}{N}\sum_{(x_i,y_i)}\ell\Bigl(S(\boldsymbol{\theta}, x_i),\,y_i\Bigr),
\end{aligned}
\end{equation}
and we take classical gradients of $\widetilde{\mathcal{L}}$ w.r.t.\ $\boldsymbol{\theta}$. This bypasses direct quantum gradients. The optimization update is:
\begin{equation}
\boldsymbol{\theta}_{t+1}
=
\boldsymbol{\theta}_t
-\eta\,\nabla_{\boldsymbol{\theta}}\widetilde{\mathcal{L}}(\boldsymbol{\theta})\bigl|_{\boldsymbol{\theta}=\boldsymbol{\theta}_t}.
\end{equation}
Because the neural surrogate can avoid the exponential gradient decay typical in quantum circuits, we do not encounter vanishing gradients from the circuit itself. We will see in experiments that this method yields stable updates even for deeper circuits.

Below is a lemma showing that if $S$ is sufficiently expressive and well-trained, then descending on $\widetilde{\mathcal{L}}$ descends on the true $\mathcal{L}$ as well:

\begin{lemma}[Surrogate Descent]\label{lemma:descent}
Assume the true loss $L(\boldsymbol{\theta})$ is continuously differentiable. Let $S(\boldsymbol{\theta})$ be a surrogate model such that $S(\boldsymbol{\theta}^{(t)}) = L(\boldsymbol{\theta}^{(t)})$ and $\| \nabla_{\boldsymbol{\theta}} S(\boldsymbol{\theta}) - \nabla_{\boldsymbol{\theta}} L(\boldsymbol{\theta})\| \le \epsilon$ for all $\boldsymbol{\theta}$ in a convex neighborhood $\mathcal{N}$ around $\boldsymbol{\theta}^{(t)}$. Then for a sufficiently small learning rate $\eta>0$, the update $\boldsymbol{\theta}^{(t+1)} = \boldsymbol{\theta}^{(t)} - \eta \nabla S(\boldsymbol{\theta}^{(t)})$ yields 
\begin{equation} L(\boldsymbol{\theta}^{(t+1)}) \le L(\boldsymbol{\theta}^{(t)}) - \frac{\eta}{2} \| \nabla L(\boldsymbol{\theta}^{(t)})\|^2 + O(\eta^2) ,\end{equation} 
provided $\boldsymbol{\theta}^{(t+1)}$ stays in $\mathcal{N}$ and $\epsilon$ is sufficiently small.
\end{lemma}

\begin{proof}
Since $L$ is differentiable, a first-order Taylor expansion around $\boldsymbol{\theta}^{(t)}$ gives:
\begin{equation}\label{eq:taylor}
    L(\boldsymbol{\theta}^{(t+1)}) = L(\boldsymbol{\theta}^{(t)}) + \nabla L(\boldsymbol{\theta}^{(t)})^\top (\boldsymbol{\theta}^{(t+1)} - \boldsymbol{\theta}^{(t)}) + R,
\end{equation}
where $R$ is the second-order remainder term: $R = \frac{1}{2}(\boldsymbol{\theta}^{(t+1)} - \boldsymbol{\theta}^{(t)})^\top H_L (\boldsymbol{\theta}^{(t+1)} - \boldsymbol{\theta}^{(t)})$ for some Hessian $H_L = \int_0^1 \nabla^2 L(\boldsymbol{\theta}^{(t)} + \tau(\boldsymbol{\theta}^{(t+1)}-\boldsymbol{\theta}^{(t)}))\,d\tau$. There exists $M>0$ such that $\|H_L\| \le M$ if $\mathcal{N}$ is bounded and $L$ is twice continuously differentiable (by the extreme value theorem). Substituting the update $\boldsymbol{\theta}^{(t+1)} = \boldsymbol{\theta}^{(t)} - \eta \nabla S(\boldsymbol{\theta}^{(t)})$ into \eqref{eq:taylor}, we get:
\begin{align}
    L(\boldsymbol{\theta}^{(t+1)}) &= L(\boldsymbol{\theta}^{(t)}) 
    - \eta\, \nabla L(\boldsymbol{\theta}^{(t)})^\top \nabla S(\boldsymbol{\theta}^{(t)}) \nonumber\\
    &\quad + \frac{\eta^2}{2}\nabla S(\boldsymbol{\theta}^{(t)})^\top H_L \nabla S(\boldsymbol{\theta}^{(t)}) \nonumber\\
    &\le L(\boldsymbol{\theta}^{(t)}) 
    - \eta\, \nabla L(\boldsymbol{\theta}^{(t)})^\top \nabla S(\boldsymbol{\theta}^{(t)}) \nonumber\\
    &\quad + \frac{\eta^2 M}{2} \|\nabla S(\boldsymbol{\theta}^{(t)})\|^2. \label{eq:L_step}
\end{align}

Using the assumption on gradient closeness: $\nabla S(\boldsymbol{\theta}^{(t)}) = \nabla L(\boldsymbol{\theta}^{(t)}) + \mathbf{e}$ with $\|\mathbf{e}\| \le \epsilon$. Then
\begin{align*}
    \nabla L(\boldsymbol{\theta}^{(t)})^\top \nabla S(\boldsymbol{\theta}^{(t)}) &= \|\nabla L(\boldsymbol{\theta}^{(t)})\|^2 + \nabla L(\boldsymbol{\theta}^{(t)})^\top \mathbf{e}\\
    &\ge \|\nabla L(\boldsymbol{\theta}^{(t)})\|^2 - \|\nabla L(\boldsymbol{\theta}^{(t)})\|\, \|\mathbf{e}\| \\
    &\ge \|\nabla L(\boldsymbol{\theta}^{(t)})\|^2 - \epsilon \|\nabla L(\boldsymbol{\theta}^{(t)})\|.
\end{align*}
Similarly, $\|\nabla S(\boldsymbol{\theta}^{(t)})\| \le \|\nabla L(\boldsymbol{\theta}^{(t)})\| + \epsilon$. Inserting these into \eqref{eq:L_step} yields:
\begin{equation}
\begin{aligned}
    L(\boldsymbol{\theta}^{(t+1)}) &\le L(\boldsymbol{\theta}^{(t)}) - \eta\Big(\|\nabla L(\boldsymbol{\theta}^{(t)})\|^2 - \epsilon \|\nabla L(\boldsymbol{\theta}^{(t)})\|\Big) \\ &+ \frac{\eta^2 M}{2}(\|\nabla L(\boldsymbol{\theta}^{(t)})\| + \epsilon)^2.
\end{aligned}
\end{equation}
For small enough $\eta$, the $\eta^2$ term is negligible compared to $\eta$ (specifically, if $\eta < \frac{\|\nabla L(\boldsymbol{\theta}^{(t)})\|}{M}$, the second-order term is $O(\eta^2 \|\nabla L\|^2)$). Ignoring $O(\eta^2)$ and using $\epsilon$ small, we get approximately:
\begin{equation} L(\boldsymbol{\theta}^{(t+1)}) \le L(\boldsymbol{\theta}^{(t)}) - \eta \|\nabla L(\boldsymbol{\theta}^{(t)})\|^2 + \eta \epsilon \|\nabla L(\boldsymbol{\theta}^{(t)})\|. \end{equation}
If we choose $\epsilon$ such that $\epsilon \ll \|\nabla L(\boldsymbol{\theta}^{(t)})\|$ in the current neighborhood (or more strictly, treat $\epsilon$ as $o(\|\nabla L\|)$), the last term becomes second-order small compared to the first term. Thus:
\begin{equation} L(\boldsymbol{\theta}^{(t+1)}) < L(\boldsymbol{\theta}^{(t)}) - \frac{\eta}{2}\|\nabla L(\boldsymbol{\theta}^{(t)})\|^2, \end{equation} 
for sufficiently small $\eta$ and $\epsilon$, completing the proof of a guaranteed decrease in $L$ up to first order.
\end{proof}

Lemma~\ref{lemma:descent} implies that as long as our surrogate $S$ captures the local slope of $L$ with reasonable accuracy (small $\epsilon$), moving opposite to $\nabla S$ will decrease the true loss $L$. In particular, even if $\nabla L$ is very small (barren plateau regime), $\nabla S$ might not be, because the surrogate can fit through sparse samples that detect slight differences in $L$. The condition $\boldsymbol{\theta}^{(t+1)}$ remains in $\mathcal{N}$ is enforceable by taking sufficiently small steps or by using a trust region radius to limit how far we go based on the surrogate's validity.

\subsection{Convergence of Surrogate Optimization}
Next, we address the convergence of the iterative surrogate-based training procedure. We will show that under the assumption that the surrogate model can be made arbitrarily accurate (by taking more sample points or a more flexible functional form as needed), the parameter sequence $\{\boldsymbol{\theta}^{(t)}\}$ produced by our method will approach a stationary point (usually a local minimum) of the true loss $L$.

\begin{theorem}[Convergence]\label{thm:convergence}
Consider the iterative update $\boldsymbol{\theta}^{(t+1)} = \boldsymbol{\theta}^{(t)} - \eta_t \nabla S^{(t)}(\boldsymbol{\theta}^{(t)})$ where $S^{(t)}$ is a surrogate model fit around $\boldsymbol{\theta}^{(t)}$ as described above. Assume: 1) The loss $L(\boldsymbol{\theta})$ is lower bounded (i.e., has a global minimum $L_{\min}$) and $L$ is $L$-Lipschitz smooth (its gradient is Lipschitz continuous with constant $L$); 2) At each iteration $t$, the surrogate $S^{(t)}$ satisfies $\| \nabla S^{(t)}(\boldsymbol{\theta}^{(t)}) - \nabla L(\boldsymbol{\theta}^{(t)}) \| \le \epsilon_t$, where $\epsilon_t$ can be made arbitrarily small by using a sufficiently accurate surrogate (and assume $\epsilon_t$ is indeed chosen small enough at each step); 3) The step size $\eta_t$ is chosen via a diminishing schedule or sufficiently small constant such that $\sum_t \eta_t = \infty$ and $\sum_t \eta_t^2 < \infty$ (a typical condition for gradient descent convergence).
Then $\{\boldsymbol{\theta}^{(t)}\}$ converges to a stationary point of $L$, i.e., $\lim_{t\to\infty} \nabla L(\boldsymbol{\theta}^{(t)}) = \mathbf{0}$, and $\lim_{t\to\infty} L(\boldsymbol{\theta}^{(t)}) = L^*$ for some local minimum value $L^* \ge L_{\min}$.
\end{theorem}

\begin{proof}
Under the Lipschitz smoothness assumption for $\nabla L$, we have the standard descent lemma for the true loss:
\begin{equation}\label{eq:lipschitz-descent}
\begin{aligned}
    L(\boldsymbol{\theta}^{(t+1)}) &\le L(\boldsymbol{\theta}^{(t)}) + \nabla L(\boldsymbol{\theta}^{(t)})^\top (\boldsymbol{\theta}^{(t+1)} - \boldsymbol{\theta}^{(t)}) \\ 
    &+ \frac{L}{2}\|\boldsymbol{\theta}^{(t+1)} - \boldsymbol{\theta}^{(t)}\|^2.
\end{aligned}
\end{equation}
Substitute the update $\boldsymbol{\theta}^{(t+1)} = \boldsymbol{\theta}^{(t)} - \eta_t \nabla S^{(t)}(\boldsymbol{\theta}^{(t)})$:
\begin{equation}
\begin{aligned}
    L(\boldsymbol{\theta}^{(t+1)}) &\le L(\boldsymbol{\theta}^{(t)}) - \eta_t\, \nabla L(\boldsymbol{\theta}^{(t)})^\top \nabla S^{(t)}(\boldsymbol{\theta}^{(t)}) \\ 
    &\quad + \frac{L\eta_t^2}{2} \|\nabla S^{(t)}(\boldsymbol{\theta}^{(t)})\|^2 \\
    &\le L(\boldsymbol{\theta}^{(t)}) - \eta_t \|\nabla L(\boldsymbol{\theta}^{(t)})\|^2 \\ 
    &\quad + \eta_t \|\nabla L(\boldsymbol{\theta}^{(t)})\|\, \|\mathbf{e}_t\| \\ 
    &\quad + \frac{L\eta_t^2}{2}(\|\nabla L(\boldsymbol{\theta}^{(t)})\| + \|\mathbf{e}_t\|)^2, \label{eq:gd-iteration}
\end{aligned}
\end{equation}
where $\mathbf{e}_t = \nabla S^{(t)}(\boldsymbol{\theta}^{(t)}) - \nabla L(\boldsymbol{\theta}^{(t)})$ and $\|\mathbf{e}_t\| \le \epsilon_t$. If $\epsilon_t$ is made sufficiently small at each iteration (ideally $\epsilon_t \to 0$ as $t\to \infty$ if needed), then \eqref{eq:gd-iteration} resembles the standard gradient descent inequality:
\begin{equation} L(\boldsymbol{\theta}^{(t+1)}) \le L(\boldsymbol{\theta}^{(t)}) - \eta_t (1 - \delta_t)\|\nabla L(\boldsymbol{\theta}^{(t)})\|^2, \end{equation} 
with some small $\delta_t$ accounting for the surrogate error and second-order term. Provided $\eta_t$ is chosen so that $0 < \eta_t L < 2$ (to satisfy the usual GD convergence conditions) and $\delta_t$ is negligible, one can show that $L(\boldsymbol{\theta}^{(t)})$ is nonincreasing and converges to a limit $L^*$. Moreover, summing \eqref{eq:gd-iteration} over $t=0$ to $T$ telescopes the differences $L(\boldsymbol{\theta}^{(t)})-L(\boldsymbol{\theta}^{(t+1)})$. This yields:
\begin{equation} \sum_{t=0}^{T} \eta_t (1-\delta_t)\|\nabla L(\boldsymbol{\theta}^{(t)})\|^2 \le L(\boldsymbol{\theta}^{(0)}) - L(\boldsymbol{\theta}^{(T+1)}). \end{equation}
As $T \to \infty$, the right-hand side is bounded by $L(\boldsymbol{\theta}^{(0)}) - L_{\min} < \infty$. If $\sum_t \eta_t = \infty$ and $\delta_t$ is bounded away from 1 (indeed $\delta_t$ can be made $0$ in ideal fitting), the only way for the left sum to remain finite is that $\|\nabla L(\boldsymbol{\theta}^{(t)})\|^2$ tends to zero sufficiently fast. Intuitively, if gradients did not tend to zero, the large number of iterations with nonzero gradient would drive $L$ below its minimum, a contradiction. Formally, using the diminishing $\eta_t$ conditions and following the standard proof of gradient descent convergence, we conclude $\lim_{t\to\infty} \nabla L(\boldsymbol{\theta}^{(t)}) = 0$.

Thus every accumulation point of the sequence $\boldsymbol{\theta}^{(t)}$ is a stationary point of $L$. Since $L$ is lower bounded and presumably our problem is well-behaved (no chaotic oscillations due to vanishing $\eta_t$ or surrogate noise), $\boldsymbol{\theta}^{(t)}$ converges to some $\boldsymbol{\theta}^*$ with $\nabla L(\boldsymbol{\theta}^*) = 0$. In practice, this will be a local minimum given the nonconvex nature of $L$. This completes the convergence proof. 
\end{proof}

Theorem~\ref{thm:convergence} ensures that our surrogate-based training will find a solution where the true gradient is zero (i.e., cannot improve $L$ further), assuming we are allowed to refine the surrogate as needed. In practice, we use a fixed surrogate model form with a moderate number of samples per iteration, which may introduce some bias $\epsilon_t$. However, as long as this bias does not systematically mislead the optimization, we observe good empirical convergence.

\section{Experimentation}

\subsection{Setup}

We evaluate on three datasets: MNIST (10 classes), CIFAR-10 (10 classes), and CIFAR-100 (100 classes). We run a noiseless statevector simulation based on the TorchQuantum library \cite{hanruiwang2022quantumnas}. The circuit has $n=15$ qubits, $4$ are ancillas at wires 3,6,9,12, with $L=6$ layers, totalling $p=7(n-1)L = 588 $ angles.

We trained all models using the AdamW optimizer \cite{loshchilov2019decoupledweightdecayregularization} with $\eta=7\times10^{-4}$ decayed by cosine scheduling, $\beta_1=0.9$, $\beta_2=0.999$, and a slight weight decay of $3\times10^{-4}$. For data loading, a batch size of $256$ is used, and various data augmentation techniques, including random cropping, random rotation, random horizontal flip, and random colour jitter, are used to prevent overfitting. All trainings were performed on a single H100 GPU for 100 epochs.

\subsection{Comparison to Existing Architectures}

We compare our surrogate QNN to a direct-gradient QNN with the same architecture but using parameter-shift rule to backpropagate through the circuit. We also compare to classical CNNs and direct NNs. The training results of different models and datasets are shown in Table~\ref{tab:mnist_cifar_results}.

\begin{table*}[ht!]
\centering
\caption{Performance on MNIST, CIFAR-10, and CIFAR-100. Accuracies are in percentage $\%$.}
\label{tab:mnist_cifar_results}
\begin{tabular}{l|ccc|ccc|ccc}
\toprule
\multirow{2}{*}{\textbf{Model}} 
& \multicolumn{3}{c|}{\textbf{MNIST}} 
& \multicolumn{3}{c|}{\textbf{CIFAR-10}} 
& \multicolumn{3}{c}{\textbf{CIFAR-100}} \\
\cmidrule{2-10}
& \textbf{Train Acc} & \textbf{Test Acc} & \textbf{Params}
& \textbf{Train Acc} & \textbf{Test Acc} & \textbf{Params}
& \textbf{Train Acc} & \textbf{Test Acc} & \textbf{Params} \\
\midrule
NN                 & 98.49$\pm$0.16 & 98.67$\pm$0.22 & 2.3M  & 73.41$\pm$0.53 & 65.92$\pm$0.65 & 2.3M  & 52.40$\pm$0.32 & 38.58$\pm$0.59 & 2.3M  \\

CNN                & 99.86$\pm$0.24 & 99.15$\pm$0.11 & 1.8M  & 98.29$\pm$0.42 & 90.22$\pm$0.48 & 1.8M  & 88.80$\pm$0.25 & 63.63$\pm$0.32 & 1.8M  \\

QNN (Direct Grad)  & 99.80$\pm$0.23 & 99.22$\pm$0.14 & 572k  & 98.30$\pm$0.29 & 89.90$\pm$0.41 & 645k  & 74.59$\pm$0.41 & 55.67$\pm$0.53 & 645k  \\

QNN (Surrogate)    & \textbf{99.96$\pm$0.05} & \textbf{99.72$\pm$0.12} & 617k  & \textbf{97.83$\pm$0.30} & \textbf{90.26$\pm$0.49} & 759k  & \textbf{68.99$\pm$0.46} & \textbf{58.65$\pm$0.36} & 759k  \\

\bottomrule
\end{tabular}
\end{table*}

\begin{figure}[h!]
\centering
\includegraphics[width=\linewidth]{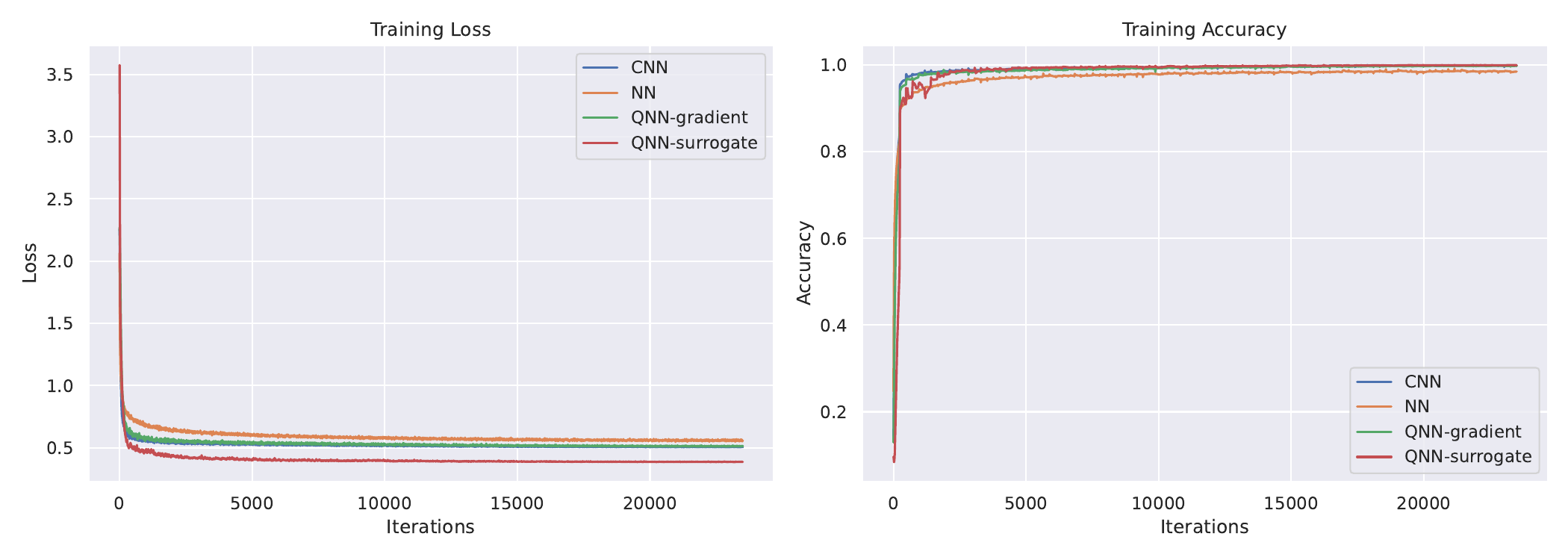}
\caption{The training loss and accuracy of the MNIST dataset for the 4 models.}
\label{fig:mnistresult}
\end{figure}

For MNIST, Table~\ref{tab:mnist_cifar_results} indicates that QNN (Surrogate) achieves near-perfect accuracy, with \(99.96\%\pm0.05\) on the training set and \(99.72\%\pm0.12\) on the test set. This surpasses both the classical CNN and the QNN (Direct Grad), and it also requires fewer parameters than the classical models. The NN baseline attains high accuracy but remains below the QNN methods. Figure~\ref{fig:mnistresult} shows the convergence process of the model during training. Overall, these observations suggest that the surrogate-based training procedure captures the data representations effectively for relatively simple classification tasks such as MNIST.

\begin{figure}[h!]
\centering
\includegraphics[width=\linewidth]{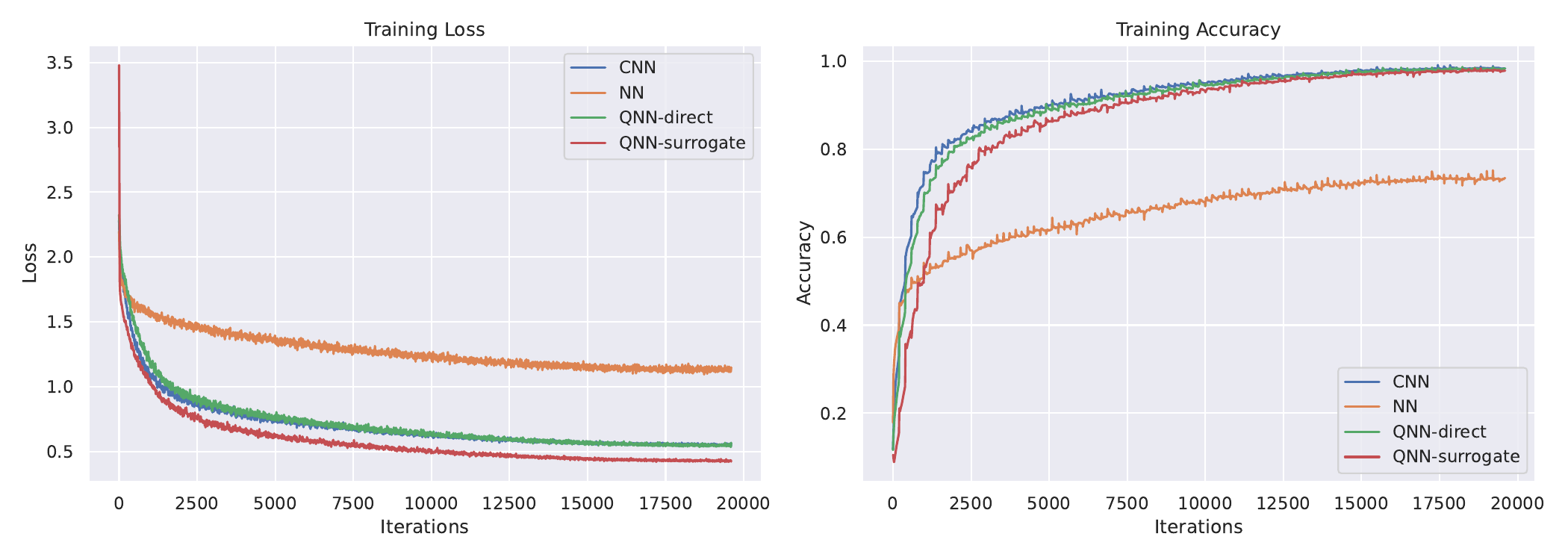}
\caption{The training loss and accuracy of the CIFAR-10 dataset for the 4 models.}
\label{fig:cifar10result}
\end{figure}

For CIFAR-10, the QNN (Surrogate) again demonstrates competitive performance, achieving \(90.26\%\pm0.49\) on the test set, which is comparable to the CNN at \(90.22\%\pm0.48\). Notably, the surrogate QNN uses approximately \(750\textrm{k}\) parameters, which is fewer than the \(1.6\textrm{M}\) parameters of the CNN. In contrast, the QNN (Direct Grad) method shows lower accuracy at \(86.31\%\pm0.41\), showing the effectiveness of the surrogate gradient approach. The simple NN achieves a test accuracy of \(65.92\%\pm0.65\) with even \(2.1\textrm{M}\) parameters, indicating that the quantum-inspired methods provide a more parameter-efficient alternative. Figure~\ref{fig:cifar10result} shows the convergence process of the model during training.

\begin{figure}[h!]
\centering
\includegraphics[width=\linewidth]{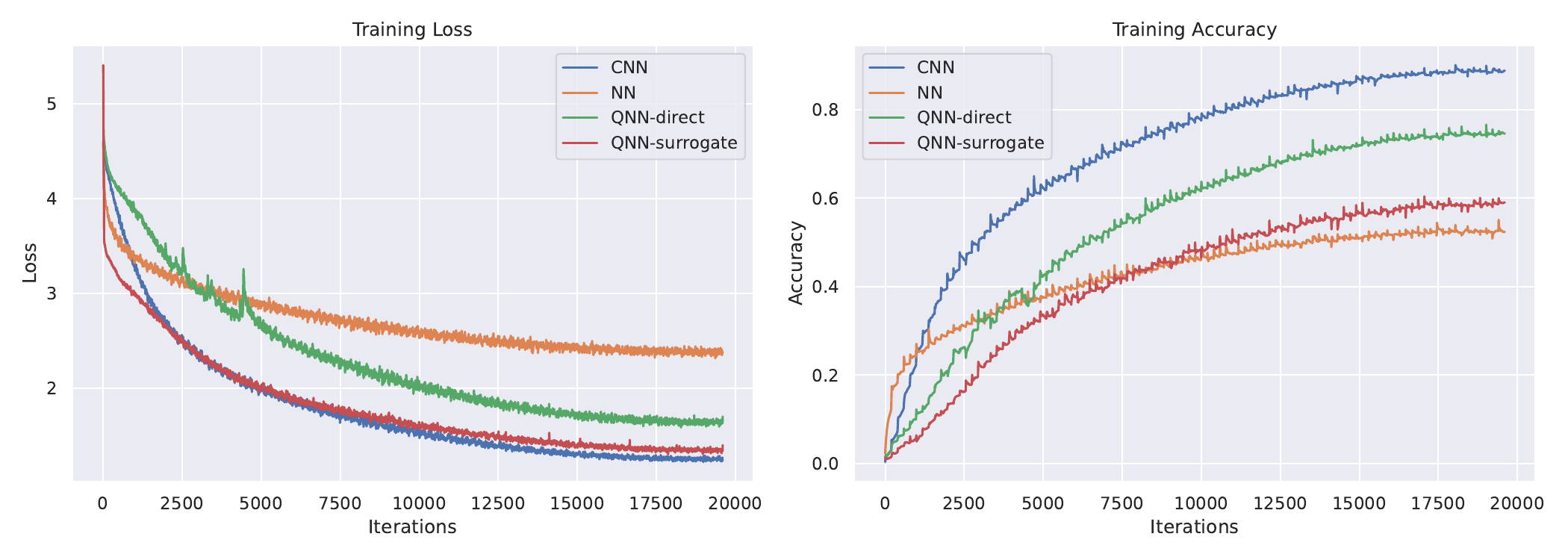}
\caption{The training loss and accuracy of the MNIST dataset for the 4 models.}
\label{fig:cifar100result}
\end{figure}

On CIFAR-100, the QNN (Surrogate) maintains a moderate level of performance at \(58.65\%\pm0.36\) on the test set, outperforming both the QNN (Direct Grad) and the basic NN. However, the CNN achieves higher accuracy (\(63.63\%\pm0.32\)), reflecting the increased difficulty of the more diverse 100-class problem. Figure~\ref{fig:cifar100result} shows the convergence process of the model during training. Nonetheless, the quantum-based approach still demonstrates a favorable trade-off in terms of parameter count (around \(750\textrm{k}\) versus the CNN’s \(1.6\textrm{M}\)), showing the potential of surrogate gradient QNNs to offer meaningful accuracy with reduced model size in more challenging classification scenarios.

\section{Discussion and Future Works}

We demonstrated a gradient-free QNN approach for image classification by combining surrogate-driven optimization with repeated mid-circuit measurement and reset operations on ancillas. This framework addresses two primary obstacles facing large-scale quantum neural networks: (i) barren plateaus causing gradients to vanish exponentially, hindering direct parameter-shift or backpropagation methods; and (ii) limited expressiveness due to purely unitary evolutions without explicit nonlinearities. Measuring and resetting ancilla qubits after each layer effectively introduces Kraus operators \(\{ K_x \}\) generating a nonunitary, potentially nonlinear transformation upon tracing out ancillas. Formally, we showed this results in a completely positive trace-preserving (CPTP) map,
\begin{equation}
\Phi(\rho) \;=\; \sum_{k} V_k\,\rho \,V_k^\dagger,
\end{equation}
where each \(V_k\) arises from interleaved unitaries \(U_\ell(\boldsymbol{\theta}_\ell)\) and projective resets that avoid collapsing main qubits into pure states. The classical surrogate, approximating the mapping \(\boldsymbol{\theta} \mapsto \mathbf{m}(\boldsymbol{\theta})\), locally models this CPTP map. Lemma \ref{lemma:descent} and Theorem \ref{thm:convergence} justify that surrogate-based parameter updates reduce the true loss \(\mathcal{L}(\boldsymbol{\theta})\) with high probability if the surrogate gradient remains close to the true gradient in a local neighborhood. Specifically, local Lipschitz smoothness of \(\mathcal{L}\) ensures
\begin{equation}
\bigl\| \nabla_{\boldsymbol{\theta}} \mathcal{L}(\boldsymbol{\theta}) \;-\; \nabla_{\boldsymbol{\theta}} S(\boldsymbol{\theta}) \bigr\|\;\le\;\epsilon,
\end{equation}
implying parameter updates guided by \(\nabla S(\boldsymbol{\theta})\) closely track true gradient descent.

A key experimental finding is that surrogate-driven training not only alleviates gradient attenuation but also yields robust classification accuracy with relatively few trainable parameters. Classical CNN baselines often employ parameter counts \(\gg 1\,\mathrm{M}\) to attain comparable performance, yet the surrogate QNN achieved \(99.72\%\) on MNIST and around \(90\%\) on CIFAR-10 with fewer than \(1\,\mathrm{M}\) parameters. On CIFAR-100, despite its complexity, the surrogate approach attained \(58.65\%\) accuracy, still using fewer parameters than classical models. This indicates that repeated ancilla resets combined with surrogate gradients enhance expressiveness while remaining resource-efficient.

Several promising future directions emerge. First, our current analysis and experiments utilized idealized simulations; therefore, a crucial next step involves assessing how realistic quantum noise (e.g., decoherence, amplitude damping, gate errors, readout noise) affects surrogate predictive power. Such noise can be modeled as additional CPTP maps, requiring an augmented surrogate \(\tilde{S}(\boldsymbol{\theta},\mathbf{n})\) with noise parameters \(\mathbf{n}\), leading to
\begin{equation}
\rho \;\mapsto\; \sum_{k}W_k(\boldsymbol{\theta},\mathbf{n})\,\rho\,W_k(\boldsymbol{\theta},\mathbf{n})^\dagger.
\end{equation}
Second, adaptive sampling strategies (e.g., active learning or Bayesian optimization) for selecting parameter configurations near \(\boldsymbol{\theta}_t\) could further enhance computational efficiency. Third, while our method targeted classification tasks, it naturally extends to other scenarios like generative modeling or quantum reinforcement learning, where differentiating quantum circuits poses challenges. Finally, deeper theoretical analysis of the function class achievable via repeated ancilla resets could validate and extend our structural insights, solidifying surrogate-driven QNNs' applicability to diverse, large-scale quantum learning problems.

\section{Conclusion}

We introduced a surrogate-driven training framework for quantum neural networks, presenting an effective alternative to traditional gradient-based methods limited by barren plateaus or computationally expensive parameter-shift rules. While demonstrated here on classification problems, the general principle—training a classical surrogate to predict quantum measurements and backpropagating through it—extends readily to other QNN tasks and architectures. Mid-circuit measurement and reset operations introduce nonunitarity, broadening the circuit’s modeling capacity. Our theoretical and experimental results affirm the accuracy, efficiency, and scalability of this surrogate approach, laying a solid foundation for future advances in quantum-classical optimization and enabling more expressive, tractable quantum neural networks for diverse applications.

\bibliographystyle{IEEEtran}
\bibliography{references}

\begin{thebibliography}{10}
\providecommand{\url}[1]{#1}
\csname url@samestyle\endcsname
\providecommand{\newblock}{\relax}
\providecommand{\bibinfo}[2]{#2}
\providecommand{\BIBentrySTDinterwordspacing}{\spaceskip=0pt\relax}
\providecommand{\BIBentryALTinterwordstretchfactor}{4}
\providecommand{\BIBentryALTinterwordspacing}{\spaceskip=\fontdimen2\font plus
\BIBentryALTinterwordstretchfactor\fontdimen3\font minus \fontdimen4\font\relax}
\providecommand{\BIBforeignlanguage}[2]{{%
\expandafter\ifx\csname l@#1\endcsname\relax
\typeout{** WARNING: IEEEtran.bst: No hyphenation pattern has been}%
\typeout{** loaded for the language `#1'. Using the pattern for}%
\typeout{** the default language instead.}%
\else
\language=\csname l@#1\endcsname
\fi
#2}}
\providecommand{\BIBdecl}{\relax}
\BIBdecl

\bibitem{mcclean2018barren}
J.~R. McClean, A.~Bohrdt, G.~S. Barron, and et~al., ``Barren plateaus in quantum neural network training landscapes,'' \emph{Nature Communications}, vol.~9, no.~1, p. 4812, 2018.

\bibitem{Cerezo2021}
M.~Cerezo, A.~Arrasmith, R.~Babbush \emph{et~al.}, ``Cost function dependent barren plateaus in shallow quantum circuits,'' \emph{Nature Communications}, vol.~12, no.~1, p. 1791, 2021.

\bibitem{Mitarai2018}
K.~Mitarai, M.~Negoro, M.~Kitagawa, and K.~Fujii, ``Quantum circuit learning,'' \emph{Physical Review A}, vol.~98, no.~3, p. 032309, 2018.

\bibitem{Preskill2018}
J.~Preskill, ``Quantum computing in the nisq era and beyond,'' \emph{Quantum}, vol.~2, p.~79, 2018.

\bibitem{Benedetti2019}
M.~Benedetti, E.~Lloyd, S.~Sack, and M.~Fiorentini, ``Parameterized quantum circuits as machine learning models,'' \emph{Quantum Science and Technology}, vol.~4, no.~4, p. 043001, 2019.

\bibitem{Schuld2015}
M.~Schuld, M.~Fingerhuth, and F.~Petruccione, ``Implementing a distance-based classifier with a quantum interference circuit,'' \emph{EPL (Europhysics Letters)}, vol. 112, no.~6, p. 60003, 2015.

\bibitem{Chertkov2022}
M.~DeCross, E.~Chertkov, M.~Kohagen, and M.~Foss-Feig, ``Qubit-reuse compilation with mid-circuit measurement and reset,'' \emph{Physical Review X}, vol.~13, no.~4, p. 041057, 2023.

\bibitem{IBM2021}
\BIBentryALTinterwordspacing
P.~Nation, ``How to measure and reset a qubit in the middle of a circuit execution,'' 2021, iBM Quantum Blog, Feb. 11, 2021. [Online]. Available: \url{https://www.ibm.com/quantum/blog/quantum-mid-circuit-measurement}
\BIBentrySTDinterwordspacing

\bibitem{Havlicek2019}
V.~Havl\'{i}\v{c}ek, A.~D. C\'{o}rcoles, K.~Temme \emph{et~al.}, ``Supervised learning with quantum-enhanced feature spaces,'' \emph{Nature}, vol. 567, no. 7747, pp. 209--212, 2019.

\bibitem{Lloyd2020}
S.~Lloyd, M.~Schuld, A.~Ijaz, J.~Izaac, and N.~Killoran, ``Quantum embeddings for machine learning,'' 2020, arXiv:2001.03622.

\bibitem{Schuld2021}
M.~Schuld, R.~Sweke, and J.~J. Meyer, ``The effect of data encoding on the expressive power of variational quantum machine learning models,'' 2021, arXiv:2101.11020.

\bibitem{householder1941}
A.~S. Householder, ``A theory of steady-state activity in nerve-fiber networks: I. definitions and preliminary lemmas,'' \emph{The Bulletin of Mathematical Biophysics}, vol.~3, no.~2, pp. 63--69, June 1941.

\bibitem{nair2010rectified}
V.~Nair and G.~E. Hinton, ``Rectified linear units improve restricted boltzmann machines,'' in \emph{Proceedings of the 27th International Conference on International Conference on Machine Learning}, ser. ICML'10.\hskip 1em plus 0.5em minus 0.4em\relax Madison, WI, USA: Omnipress, 2010, p. 807–814.

\bibitem{hanruiwang2022quantumnas}
H.~Wang, Y.~Ding, J.~Gu, Z.~Li, Y.~Lin, D.~Z. Pan, F.~T. Chong, and S.~Han, ``Quantumnas: Noise-adaptive search for robust quantum circuits,'' in \emph{The 28th IEEE International Symposium on High-Performance Computer Architecture (HPCA-28)}, 2022.

\bibitem{loshchilov2019decoupledweightdecayregularization}
\BIBentryALTinterwordspacing
I.~Loshchilov and F.~Hutter, ``Decoupled weight decay regularization,'' 2019. [Online]. Available: \url{https://arxiv.org/abs/1711.05101}
\BIBentrySTDinterwordspacing

\end{thebibliography}

\end{document}